\documentclass[11pt]{article}
\usepackage[margin=2cm, top=2cm, bottom=2.5cm,a4paper]{geometry}
\usepackage{bbding}
\usepackage{authblk}
\usepackage{bbm}
\usepackage{comment} 
\usepackage{fancyhdr} 
\usepackage{fancyvrb} 
\usepackage{graphicx} 
\usepackage{listings} 
\usepackage{mathtools}
\usepackage[version=3]{mhchem}
\usepackage{caption, threeparttable}
\usepackage{paralist} 
\usepackage{url} 

\usepackage[english]{babel}

\usepackage{lmodern}
\usepackage{amsmath}
\usepackage{amsthm}
\usepackage{amssymb}
\usepackage[hang, flushmargin]{footmisc}
\usepackage{hyperref}
\usepackage{footnotebackref}
\usepackage{setspace}
\usepackage{pdflscape}
\usepackage{xcolor}
\usepackage{pdflscape}
\usepackage{caption}
\usepackage{float}
\usepackage{subcaption}

\usepackage{newtxtext}

\def\be{\begin{equation}}
\def\ee{\end{equation}}
\def\bes{\begin{equation*}}
	\def\ees{\end{equation*}}
\def\bea{\begin{equation} \begin{aligned}}
\def\eea{\end{aligned} \end{equation}}
\def\beas{\begin{equation*} \begin{aligned}}
		\def\eeas{\end{aligned} \end{equation*}}

\def\bi{\begin{itemize}}
	\def\ei{\end{itemize}}

\newcommand{\R}{\mathbb{R}}

\newcommand{\E}{\mathbb{E}}

\newcommand{\XXX}{\mathcal{X}}
\newcommand{\BBB}{\mathcal{B}}

\newcommand{\GGG}{\mathcal{G}}

\newcommand{\LLL}{\mathcal{L}}

\newcommand{\xx}{\mathbf{x}}
\newcommand{\yy}{\mathbf{y}}

\newcommand{\FF}{\mathbf{F}}
\newcommand{\GG}{\mathbf{G}}
\newcommand{\HH}{\mathbf{H}}

\newcommand{\dd}{\mathrm{d}}

\newcommand{\lp}{\left(}
\newcommand{\rp}{\right)}

\newcommand{\beq}{\begin{equation}}
\newcommand{\eeq}{\end{equation}}
\newcommand{\beqs}{\begin{subequations}}
\newcommand{\eeqs}{\end{subequations}}
\newcommand{\balign}{\begin{align}}
\newcommand{\ealign}{\end{align}}

\newcommand{\benorfig}{\begin{figure}\centering\begin{measuredfigure}}
\newcommand{\enorfig}{\end{measuredfigure}\end{figure}}

\newtheorem{remark}{Remark}[section]

\newtheorem{proposition}{Proposition}[section]

\begin{document}

\title{On Some Aspects of the Response to \\ Stochastic and Deterministic Forcings}

\author[1,2,3]{Manuel Santos Guti\'errez \thanks{Email: \texttt{manuel.santos-gutierrez@weizmann.ac.il}}}
\author[1,2]{Valerio Lucarini\thanks{Corresponding author. Email: \texttt{v.lucarini@reading.ac.uk}}} 
\affil[1]{Department of Mathematics and Statistics, University of Reading, Reading, U.K.}
\affil[2]{Centre for the Mathematics of Planet Earth, University of Reading, Reading, U.K.}
\affil[3]{Department of Earth and Planetary Sciences, Weizmann Institute of Science, Rehovot 76100, Israel}

\maketitle

\begin{abstract}
 The perturbation theory of operator semigroups is used to derive response formulas for a variety of combinations of acting forcings and reference background dynamics. {\color{black}In the case of background stochastic dynamics, we decompose the response formulas using the Koopman operator generator eigenfunctions} and the corresponding eigenvalues, thus providing a functional basis towards identifying relaxation timescales and modes in physically relevant systems. To leading order, linear response  gives the correction to expectation values due to extra deterministic forcings acting on either stochastic or chaotic dynamical systems. When considering the impact of weak noise, the response is linear in the intensity of the (extra) noise for background stochastic dynamics, while the second order response given the leading order correction when the reference dynamics is chaotic. In this latter case we clarify that 
 previously published diverging results can be brought to common ground when a suitable interpretation--- Stratonovich vs. Ito--- of the noise is given. 
 Finally, the response of two-point correlations to perturbations is  studied through the resolvent formalism via a perturbative approach. Our results allow, among other things, to estimate how the correlations of a chaotic dynamical system changes as a results of adding stochastic forcing.
\end{abstract}

\tableofcontents

\section{Introduction}

Response theory aims at predicting the how the statistical properties of a system under investigation changes as a result of the application of (weak) external stimuli. In linear approximation, the cornerstone result to this end is the fluctuation dissipation theorem (FDT) \cite{kubo,kubo1966}, which relates the statistics of the unforced fluctuations and the relaxation properties of the system with its response to external forcings and,  specifically allows for writing the Green function describing the response of the system as a correlation between observables of the unperturbed state. Hence, the FDT classically provides a powerful tool for predicting the properties of forced fluctuations from those of the unforced ones. See \cite{Lacorata2007,marconi2008,Sarracino2019} for a comprehensive summary of the main results and more modern perspectives, and \cite{LucariniColangeli2012} for an extension of the theory to the nonlinear case. 

{\color{black}The applicability of the FDT relies on the existence of a smooth (with respect to Lebesgue) invariant measure, which allows for applying an integration by parts in the derivation of the final formulas. In the case of deterministic dynamics, this is  viable only for systems with vanishing average phase space contraction, as in the case of physical system at equilibrium.} When  non-conservative forces come into play, phase space volume contracts, making the invariant measure describing the statistical steady state singular with respect to the Lebesgue measure. In this case, one needs to resort to  more general response formulas. Ruelle \cite{ruelle_nonequilibrium_1998,ruelle2009} introduced a response theory able to predict the change in the statistical properties of Axiom A systems due to small perturbation to its dynamics. Directly implementing the Ruelle response formulas is particularly nontrivial in the case of deterministic chaotic systems, because of the radically different properties of the tangent space in its unstable and stable directions \cite{abramov2007}. {\color{black}Fortunately, recent contributions based on adjoint and shadowing methods seem to provide a convincing way forward  \cite{Wang2013,Chandramoorthy2020,Ni2020,Sliwiak2022}}. Ruelle's response theory has been then reformulated and extended through the viewpoint of functional analysis by studying via a perturbative expansion the corrections to the transfer operator \cite{B00} due to the applied forcing \cite{liverani2006,BL07,B08}. {\color{black}Systems of stochastic differential equations (SDEs) \cite{Khasminskii1966,kloeden_1992,arnold1998}} typically posses smooth invariant probability densities, which ensures the applicability of the FDT and the straightforward  derivation of linear response formulas \cite{Hanggi1982,Hairer2010,pavliotisbook2014}. {\color{black}See also the the recent work by \cite{dieball_2022}, where the Feynman-Kac formalism is used to derive expressions for the covariances of general additive dynamical functionals for stochastic dynamical systems, and the comprehensive review by \cite{maes_2020aa}, where a trajectory-based approach on the computation of the response focusing on  computing and interpreting the time-symmetric and time-antisymmetric contributions is discussed in detail.} 


Linear response theory has been used  to study sensitivity and transport properties of a plethora of physical problems and systems, including  stochastic resonance \cite{gammaitoni1998}, optical materials \cite{Lucarini2005},  simple toy models of chaotic dynamics \cite{reicklinear2002,Cessac2007,lucarini2009b,Lucarini2011}, {\color{black}Markov chains \cite{mitrophanov_2004,Lucarini2016,Antown2018,SantosGutierrez2021}},  neural networks \cite{Cessac2019},  turbulence \cite{Kaneda_020},  galactic dynamics \cite{BinneyTremaine2008},  financial markets \cite{Puertas2021},  plasma \cite{Nambu1876}, interacting identical agents \cite{Lucarini2020PRSA,Zagli2021,Amadori2021}, optomechanical systems \cite{Motazedifard2021}, and the climate system \cite{Leith1975,North1993,Lucarini2017,Lembo2020,Bastiaansen2021}, just to name  a few; see also the many examples discussed in \cite{Ottinger2005}, the  theoretical developments presented in \cite{lucarini2008,Baiesi2013,Barbier2018,Lucarini2018JSP,Sarracino2019,Wormell2019}, and the recent special issue edited by Gottwald \cite{Gottwald2020}.


When considering SDEs, the theory of Markov semigroups and of their spectral decomposition (see \cite{Tantet2018b,chekroun2019c}) can be applied to explain the properties of correlation functions. 
In the formal limit of weak noise, the spectral properties resemble those of the underlying deterministic component (cf. \cite{kellerliverani1999}) but, as mentioned above, in this case formal problems emerge when trying to relate correlations with response formulas. 

{\color{black}The first goal of this note is to study the linear response of SDEs to perturbations that can affect either the drift term or the stochastic forcing.} 
Our strategy will rely on investigating the impact of applying perturbations to the dynamics in terms of corrections to the Fokker-Planck equation or to the backward-Kolmogorov equation, which decribe the evolution of the smooth measures and observables of the system, respectively. This will leads us to computing the Green functions of the system and subsequently express it in terms of Koopman modes and the corresponding eigenvalues \cite{Mezic2005,Budinisic2012}{\color{black}, taking advantage of the spectral theory of Markov semigroups.} The basic derivations are revisited here in Section~\ref{sec:time dependent forcing}. {\color{black}The advantage of taking this angle on the problem is that one can relate the relaxation and oscillatory time scales of the Koopman modes, which are intrinsic to the unperturbed system, with the time-dependent response of general observables for general forcings. Additionally, we derive similalry general explicit expressions for the spectral properties of the frequency-dependent response of the system. Such a connection is particularly illuminating in the case of systems whose Koopman modes possess a clear distinction between the point and the residual spectrum \cite[Theorems 4-6]{SantosGutierrez2021,chekroun2019c}, as discussed in detail below.}

We will also explore another aspect of the relationship between the response of a system to perturbations and stochasticity. Stochastic perturbations to deterministic systems have been widely studied in the literature \cite{kifer1988,freidlin1998} and the investigation of the weak-noise limit in the case the dynamics of the underlying system is chaotic is of key importance for constructing the so-called physical measure \cite{Eckmann1985,ruelle1989}. The reader is also referred to \cite[Chapter 7]{gardiner2009} for an introduction to the impact of stochastic perturbations and to \cite{Blank1998,kellerliverani1999} for abstract results concerning the transfer operator. In Section~\ref{sec:stochastically forced} we shall, nonetheless, follow previous work based on response theory whereby, taking advantage of Ruelle's formalism, the effects of adding noise to a deterministic system are computed \cite{Lucarini2012}. The same question is tackled in \cite{Abramov2017}, although the author resorts to expanding the perturbed flow in terms of the tangent map leading to an apparently different formula to that found in \cite{Lucarini2012}. The relationship between these two formulas will be discussed below in Proposition~\ref{prop:ito stratonovich}, elucidating the importance of considering the correct interpretation for the noise term. 

There is a somewhat surprising lack of studies focusing on the systematic investigation of the response of higher-order correlations to forcings. As shown in \cite{lucarini2017b}, this problem is instead of relevance also for understanding the sensitivity to coarse-grained representations of the system of interest. {\color{black}As additional output of our work,} in Section~\ref{sec:response correlations} we clarify that it is possible to extend the results presented in \cite{lucarini2017b} to the case of perturbations acting on either the drift of the noise law of a system, whether its background dynamics is deterministic or stochastic. We then achieve a possibly useful extension of the response formulas able to encompass many cases of interest where one wishes to  study the sensitivity of the correlations of a system to perturbations. 

Finally, in Section \ref{sec:general remarks} we draw our conclusions and present perspectives for future work. Appendix~\ref{sec:homogeneous equation} revisits the homogeneous equation for the response derived originally in \cite{kenkre1971} where, here, we rewrite such equation in terms of the stochastic analysis language employed throughout this note. 

\section{Semigroups and Response}\label{sec:time dependent forcing}

In this section, the construction of the Green function for an SDE is revisited using Fokker-Planck operator expansions. We start by considering the $d$-dimensional It\^o SDE with deterministic drift $\FF : \R^d \longrightarrow \R^d$, perturbation field $\GG:\R^d \longrightarrow \R^d$, modulated by the bounded function $g: \R \longrightarrow \R$, {\color{black}and noise law given by the $d\times p$ matrix $\Sigma \in \R^{d\times p}$, $p\geq 1$}:
 \begin{equation}\label{eq:sto ode 2}
\dd \xx (t) = \left[\FF (\xx)  + \varepsilon g(t)\GG(\xx) \right]\dd t+\Sigma(\xx)\dd W_{t},
\end{equation}
where $W_t$ is a {\color{black}$p$}-dimensional Wiener process and $\varepsilon$ is a real number. We assume that the SDE \eqref{eq:sto ode 2} generates a process $\xx(t)$ in $\R^d$ for an initial condition at a certain value of time. The SDE is is autonomous when $\varepsilon=0$. 
When $\varepsilon \neq 0$, a time-modulated perturbation vector field $\GG$ is activated through the  function $g$. The evolution of probability density functions associated with the SDE \eqref{eq:sto ode 2}, is given by the Fokker-Planck equation:
\begin{equation}\label{eq:fpe 2}
\partial _t\rho =\LLL(t) \rho= -\nabla \cdot \lp \FF\rho \rp  - \varepsilon g(t)  \nabla \cdot \lp \GG\rho \rp +\frac{1}{2}\nabla^2:\lp  \Sigma \Sigma^{\top}\rho \rp,
\end{equation}
where the $\Sigma \Sigma^{\top}$ is the noise covariance matrix and 
the operator $\LLL$ is divided into two operators: $\LLL_0:D(\LLL_0)\subseteq \mathcal{B} \longrightarrow \mathcal{B}$ and $\LLL_1:D(\LLL_1)\subseteq \mathcal{B} \longrightarrow \mathcal{B}$ so that $\LLL(t) = \LLL_0 +\varepsilon g(t) \LLL_1 $ and their domains, denoted by $D(\cdot)$, are densely defined in a suitable Banach space $\BBB$ and are assumed to satisfy $D(\LLL (t))=D(\LLL_0)=D(\LLL_1)$. These operators are defined as:
\begin{subequations}
	\begin{align}
	\LLL_0 \rho&= -\nabla \cdot \lp \FF \rho \rp + \frac{1}{2}\nabla ^2 : \lp \Sigma\Sigma^{\top}\rho \rp  ;\label{eq:pert op fokker planck equation 1} \\
	\LLL_1 \rho &= -\nabla \cdot \lp \GG \rho \rp \label{eq:pert op 1},
	\end{align}
\end{subequations}
for every $\rho $ in  $D(\LLL_0)$. We immediately identify that $\LLL_1$ is  the operator corresponding to the perturbation introduced in Eq.~\eqref{eq:sto ode 2} and hence, we refer to it as the perturbation operator. Because the right-hand side (RHS) of Eq.~\eqref{eq:sto ode 2} is linear, its solutions when $\varepsilon = 0$ are given by the exponential map family $\{ e^{t\LLL_0}\}_{t\geq0}$, which we call the semigroup generated by $\LLL_0$. In fact, $\LLL_0$ generates a strongly continuous semigroup so that the properties of the exponential function are inherited \cite{engel2000}. It is noted, however, that $\LLL$ is time-dependent, so the way it generates an evolution semigroup is by means of the \emph{time-ordered} exponential operator \cite{gill2017}. 

The regularity of the transition probabilities in Eq.~\eqref{eq:fpe 2} and encoded in $\mathcal{B}$ is not granted by merely adding noise, but one has to ensure that the latter spreads in all directions. To this end, H\"ormander's condition \cite{Hoermander.1967} is invoked so that at any point in phase space $\XXX$, the tangent linear is recovered by the directions generated by the second order differential operator in $\LLL_0$ stemming from the noise. Hence, the smoothness and boundedness of the coefficients of Eq.~\eqref{eq:sto ode 2} together with H\"ormander's condition ensure the existence of smooth transition probabilities solving Eq.~\eqref{eq:fpe 2} and an invariant probability density function $\rho_0$ that solves $\LLL_0\rho_0\equiv0$ \cite{pavliotisbook2014}. 

Because the SDE is non-autonomous, the density functions $\rho$ are dependent on the phase variable $\xx$ at time $t$ and the initial condition $\xx_0$ at time $s$. If one, instead, assumes that the initial condition $\xx_0$ at time $s=0$ is distributed according to the unperturbed state $\rho_0$ ($\varepsilon=0$), we can understand the densities in Eq.~\eqref{eq:fpe 2} as
\begin{equation}
\rho(\xx,t):=\int \rho(\xx,\yy,t,0)\rho_0(\yy)\dd \yy,
\end{equation}
so that $\rho(\xx,t)$ determines the probability of encountering the process in $\xx$ at time $t$ for a $\rho_0$-distributed initial condition; see also \cite[Chapter 2]{pavliotisbook2014} for more general considerations. 

The goal now is to calculate the distribution function for the perturbed system i.e., when $\varepsilon \neq 0$. To this end, we shall solve Eq.~\eqref{eq:fpe 2} for successive orders of $\varepsilon$. Let us assume that a solution of Eq.~\eqref{eq:fpe 2} $\rho$ can be written as:
\begin{equation}\label{eq:expansion time dependent density}
\rho_\varepsilon = \rho_0 + \varepsilon \rho_1 + \mathcal{O}(\varepsilon^2)
\end{equation}
{\color{black}Such an asymptotic expansion is the starting point of virtually all linear response formulas for statistical mechanical systems. We assume here that our systems of interest comply with the conditions discussed in \cite{Hairer2010}. A discussion of the radius of convergence of the expansion in the context of finite-state Markov processes is presented in \cite{Lucarini2016,SantosJSP}.} 

While $\rho_0$ is provided, $\rho_1$ is obtained by plugging Eq.~\eqref{eq:expansion time dependent density} into Eq.~\eqref{eq:fpe 2} and gathering, first, the $\varepsilon$-order terms and applying the variation-of-parameters formula:
\begin{equation}
\rho_1(\cdot,t) = \int_{0}^te^{(t-s)\LLL_0}g(s)\LLL_1\rho_0 \dd s,\label{rho1}
\end{equation}
where the time-modulation $g$ appears in the integrand. The $k^{th}$ order element of Eq.~\eqref{eq:expansion time dependent density} is obtained analogously by gathering $\varepsilon^k$-terms \cite{kubo,ruellegeneral1998,lucarini2008}. {\color{black}The expectation value of a general observable $\Psi$  can be written as}:
\begin{align}
	\left \langle \Psi , \rho(\cdot,t) \right \rangle &=\sum_{k=0}^{\infty}\delta^{(k)}\left[\Psi \right](t)= \int \Psi(\xx) \rho_0(\xx)\dd \xx+ \varepsilon \int \Psi(\xx) \rho_1(\xx, t)\dd \xx + \mathcal{O}\lp \varepsilon^2 \rp. \label{eq:response function a}
\end{align}
Hence, after integrating by parts and employing the dual representation with ``$\ast$'', the \emph{linear response} accounting for the first order corrections is:
\begin{equation}\label{eq:linear response time dependent}
\delta^{(1)}[\Psi] (t) = \varepsilon \int_{-\infty}^{\infty} g(s)\int \Theta(t-s)e^{(t-s)\LLL^{\ast}_0}\Psi(\xx) \LLL_1\rho_0(\xx) \dd \xx \dd s=\varepsilon  \lp g \ast \GGG \rp (t).
\end{equation}
{\color{black}where $\Theta(t)$ is the Heaviside distribution}. Specifically, the generator of the unperturbed backward-Kolmogorov equation  $\LLL_0^\ast$ and the analogous perturbation operator $\LLL_1^{\ast}$ act on measurable observables $\Psi$ and can be written as 
\begin{subequations}
\begin{align}
\LLL_0^\ast \Psi&= \FF \cdot \nabla \Psi  + \frac{1}{2}\Sigma\Sigma^{\top}: \nabla ^2  \Psi ; \label{L0aststo} \\
\LLL_1^\ast \Psi &= \GG \cdot \nabla \Psi.\label{L1asta}
\end{align}
\end{subequations}
Furthermore, the Green function $\GGG (t)=\Theta(t)\int e^{t\LLL^{\ast}_0}\Psi(\xx)\LLL_1 \rho_0(\xx) \dd \xx$ has a non-negative support because of \emph{causality} \cite{ruelle2009,Lucarini2018JSP}. 
Equation \ref{eq:linear response time dependent} implies that once the Green function is known or computed off-line from the unperturbed regime, the linear response can be readily calculated regardless of the time-modulation $g$.
\begin{remark}
It is also possible to study the impact of applying a perturbation to the stochastic term in Eq.~\eqref{eq:sto ode 2} as follows:
\begin{equation}\label{eq:sto ode 3}
\dd \xx (t) = \FF (\xx)  \dd t+\Sigma(\xx) \dd W_{t}+ \varepsilon g(t)\Gamma (\xx)\dd W_{t},
\end{equation}
{\color{black}where $\Gamma\in\mathbb{R}^{d\times p}$ is a $d\times p$ matrix defining the correction to   
the background noise law. 
Note that we allow for the presence of nontrivial time modulations in the properties of the noise through the bounded function $g(t)$.} Up to first order in $\varepsilon$, the Fokker-Planck equation reads as follows:  
\begin{equation}\label{eq:fpe 2b}
\partial _t\rho =\LLL(t) \rho= -\nabla \cdot \lp \FF\rho \rp  - \varepsilon g(t)  \frac{1}{2}\nabla ^2 : \lp \lp \Sigma\Gamma^{\top} + \Gamma\Sigma^{\top} \rp\rho \rp +\frac{1}{2}\nabla^2:\lp  \Sigma \Sigma^{\top}\rho \rp.
\end{equation}
In this case, one has $\mathcal{L}(t)=\mathcal{L}_0+\varepsilon g(t) \mathcal{L}_1$ where 
$\LLL_1$ is:
\begin{equation}
    \LLL_1 \rho = \frac{1}{2}\nabla ^2 : \lp \lp \Sigma\Gamma^{\top} + \Gamma\Sigma^{\top} \rp\rho \rp.\label{L1stoch}
\end{equation}
Hence, by plugging this expression of $\LLL_1$ into Eqs.~\eqref{rho1}-\eqref{eq:linear response time dependent} one can compute the first order correction to the statistical properties of the system due to perturbations to the stochastic diffusion. The same formulas--- apart from the specific expression of $\LLL_1$--- can be used to evaluate at leading the impact of either deterministic or stochastic forcing onto the system. Additionally, because of linearity, the impact of deterministic (as in Eq.~\eqref{eq:sto ode 2}) and stochastic perturbations (as in Eq.~\eqref{eq:sto ode 3}) can be treated separately and then added up to account for the response of the system to a combined perturbation.

Note that if $\Sigma \Gamma^\top + \Sigma^\top \Gamma$ vanishes (and in particular if $\Sigma\equiv 0$, i.e. if the background dynamics is deterministic) the leading order response to a stochastic perturbation is proportional to $\varepsilon^2$, see \cite{Lucarini2012,Abramov2017} and discussion in Section~\ref{sec:stochastically forced}. Note also that the  {\color{black}dual operator $\LLL_1^\ast$} acting on measurable observables $\Psi$ can be written as: 
\begin{equation}
\LLL_1^\ast \Psi= \frac{1}{2}\lp \Sigma\Gamma^{\top} + \Gamma\Sigma^{\top} \rp : \nabla ^2 \Psi.\label{L1astb}
\end{equation}
{\color{black}Additionally, note also that if we include in Eq.~\eqref{eq:sto ode 3} extra noise terms of the form $\varepsilon f(t) \Delta(\xx)\mathrm{d}V_t$, where  $\Delta\in\mathbb{R}^{d\times q}$ is a $d\times q$ matrix, $f$ is a bounded function, and $V_t$ is a $q$-dimensional Wiener process with no correlation with the previously defined $W_t$  $p$-dimensional Wiener process, their leading-order impact on the measure is $\mathcal{O}(\varepsilon^2)$, because the first order term of Eq.~\eqref{L1stoch} and \eqref{L1astb} would be zero.}    

\end{remark}
\subsection{Spectral Decomposition of the Linear Response}\label{sec:spectral time dependent}
Under mild conditions on the SDE~\eqref{eq:sto ode 2} and its associated Fokker-Planck equation~\eqref{eq:fpe 2}, the semigroup $\{e^{t\LLL_0}  \}_{t\geq 0}$ enjoys an exponential convergence to equilibrium due to $\LLL_0$ having a nontrivial set of dominating eigenvalues \cite[Theorems 4-6]{chekroun2019c}. The target of this section is, thus, to decompose the Green function $\GGG$ in terms of the spectrum of the unforced generator $\LLL_0$ of the Fokker-Planck equation in order to extract the relaxation timescales of the response. For this, we shall further require the semigroup $\{e^{t\LLL_0}\}_{t\geq0}$ to be \emph{quasi-compact}, this is, it approaches the space of compact operators as $t$ tends to infinity in the operator norm \cite{engel2000}. Quasi-compact semigroups are constituted by operators having a strictly negative essential growth bound $r_{ess}$\footnote{The essential growth bound is defined as:\begin{equation*}
	r_{ess} = \inf_{t>0} \frac{1}{t} \log \left\| e^{t\LLL_0} \right\|_{ess},
\end{equation*} where $\|\cdot \|_{ess}$ is the distance to compact operators in the operator norm.}, implying that spectra with modulus larger than $e^{r_{ess}t}$ are eigenvalues of finite algebraic multiplicity, which are related to those of the generator by means of the Spectral Mapping Theorem \cite[Chapter IV]{engel2000}. Namely, if $\lambda_j$ is an eigenvalue of $\LLL_0$ with eigenfunction $\psi_j$, so is $e^{\lambda_jt}$ of $e^{t\LLL_0}$ relative to the same eigenfunction. Therefore, if $\{ \lambda_j \}_{j=1}^{M}$ are $M$ eigenvalues of finite algebraic multiplicity $a_j$ such that $\lambda_j > r_{ess}$ the evolution operator $e^{t\LLL_0}$ is decomposed as:
\begin{equation}\label{Eq_L0_decomp chapter 2}
e^{t\LLL_0} = \sum _{j=1}^{ M} T_j(t) + \mathcal{R}(t),
\end{equation}
where $T_{j}(t)$ denotes the contribution relative to the eigenvalue $\lambda_j$ and $\mathcal{R}$ is the operator accounting for the essential spectrum. The operator $T_j(t)$ is defined as:
\begin{align} \label{eq:spectral projector}
T_j(t)=\sum_{k=0}^{a_j-1}e^{\lambda_jt}\frac{t^k}{k!}\lp \LLL_0 - \lambda_j
\rp^{k}\Pi _j,   \end{align}
where $\Pi _j$ denotes the spectral projector corresponding to the eigenvalue $\lambda_j$. {\color{black}It is worth stressing that Eq.~\eqref{eq:spectral projector} applies also to non-diagonalisable operators. Indeed, in case the eigenvalue $\lambda_j$ has non-unit multiplicity, i.e. $a_j>1$, Eq.~\eqref{eq:spectral projector} accounts for the generalised eigenfunctions and possible ``Jordan blocks''. } For notational convenience, we shall further assume that $\lambda_0=0 > \mathfrak{Re}\lambda_1\geq \mathfrak{Re}\lambda_2\geq \ldots  $. The hope is that, as $t$ tends to infinity, the contributions from the essential spectrum decay. In fact, if $r>\sup\{r_{ess}\}\cup \big\{ \mathfrak{Re}\lambda : \lambda \in \sigma(\LLL_0)\setminus \{\lambda_0,\ldots,\lambda_M\} \big\}$ there exists a constant $C>0$ such that:
\begin{equation}\label{eq:bound residual}
\| \mathcal{R}(t) \| \leq Ce^{r t},
\end{equation}
for all positive values of $t$. This ensures that, as time goes to infinity, the norm of the residual operator decays to zero. The finiteness of the eigenvalues assumed here implies that all of them, but for $\lambda_{0}=0$, are contained in a strip of the complex plane that is at a distance $\gamma=-\mathfrak{Re}\lambda_{1}$ from the imaginary axis. The number $\gamma$ is called the \emph{spectral gap} of $\LLL_0$ and it guarantees that correlations decay exponentially fast \cite[\S 1.3]{chekroun2019c,B00}. 
It must be noted, though, that the point spectrum need not be, in general, finite, so that it is possible that it accumulates around the leading zero eigenvalue. {\color{black}This pathological case is responsible for subexponential decay rates \cite{Ruelle1986} and a signal of the system approaching a critical point, as is the case of the pitchfork bifurcation \cite{gaspard1995} or the critical transitions between competing attractors in atmospheric circulation models \cite{Tantet2018b}.}

The spectral decomposition of quasi-compact operators allows us to naturally link the theory of the Green function with the eigenvalues of the generator of the unperturbed Fokker-Planck semigroup $\LLL_0$. {\color{black}We remark that in the case of smooth invariant measure, such eigenvalues agree with those of the generator of the backward-Kolmogorov equation $\LLL^\ast_0$}. {\color{black}If we assume for the sake of  simplicity that our observable of interest $\Psi$ projects entirely onto the point spectrum of $\LLL^\ast_0$, we obtain the following closed formula:}
\begin{subequations}\label{eq:eigenvalues linear response time dependent}
	\begin{align}
		\delta^{(1)}\left[ \Psi\right](t)&=\varepsilon \lp g \ast \GGG \rp(t)= \varepsilon \int_{-\infty}^{\infty} g(t-s)\GGG(s) \dd s \\&=\varepsilon \int_{-\infty}^{\infty} g(t-s)\int  \Theta(s)e^{s\LLL_0^{\ast}}\Psi(\xx)\LLL_1\rho_0(\xx) \dd \xx \dd s \\ &=\varepsilon \int_{-\infty}^{\infty} g(t-s)\int \lp \Theta(s) \sum_{j=1}^M e^{\lambda_js}\Pi^{\ast} _j \Psi(\xx) \rp \LLL_1\rho_0(\xx) \dd s\dd \xx \\&= \varepsilon \sum_{j=1}^{M}\sum_{k=1}^{a_j-1} \alpha_j^{(k)}\frac{1}{k!} \int_{-\infty}^\infty g(t-s)\Theta(s)e^{\lambda_js}s^k\dd s \\&= \varepsilon \sum_{j=1}^{M}\sum_{k=1}^{a_j-1} \alpha_j^{(k)}\frac{1}{k!} \left[ g \ast \lp \Theta(\circ) e^{\lambda_j \circ }\circ^k\rp \right](t),
	\end{align}
\end{subequations}
where the coefficients $\alpha_j^{(k)}$ are defined as:
\begin{equation}\label{eq:alphas definition}
\alpha^{(k)}_j=\int (\LLL^{\ast}_0 - \lambda_j)^k\Pi^{\ast}_j\Psi (\xx)\LLL_1 \rho_0(\xx)\dd\xx.
\end{equation}
and $\Pi^\ast_j$ is the projector of the observable $\Psi$ on the $j^{th}$ mode of the generator $\LLL_0^{\ast}$. The $\alpha$ coefficients  weight the contribution of the various modes of the response for a given combination of observable and forcing. 

{\color{black}The calculations presented in Eq.~\eqref{eq:eigenvalues linear response time dependent} assume that the observable of interest $\Psi$ projects entirely onto the point-spectrum of $\LLL_0$, but in general, $\Psi$ can possess a significant component contained in the residual spectrum. The quasi-compact assumption, however, ensures that after $1/r$ time units, the residual contribution to the response $\delta^{(1)}[\Psi](t)$ will be negligible, where $r$ was introduced in Eq.~\eqref{eq:bound residual}. In essence, the further away the residual spectrum is from the imaginary axis, the less significant will its contribution be to the response as a function of time. Hence, the reason for neglecting the residual spectrum is twofold. First, the calculations become exact and, second, we emphasize the role of the dominating eigenvalues in determining the decay rates of the response function.}

As clear from the discussion above, these results apply to either a deterministic or a stochastic perturbation to the dynamics of an SDE  or to a combination thereof; one just needs to use the right definition for the operator $\LLL_1$. Notice that the contribution of the eigenvalue $0$ is not present because it is simple and its associated dual eigenfunction is constant, making $\alpha_0 = 0$. We understand here that the spectral gap gauges the sensitivity of the system with respect to external perturbations: if $0<\gamma \ll1$ the linear response  $\delta^{(1)}[\Psi](t)$ to an impulse defined by $g(t)=\delta(t)$, where $\delta(t)$ is Dirac's delta, will decay very slowly to zero for any observable $\Psi$. In physical terms, this means that the negative feedbacks of the system are very weak.

At a practical level, when a system is weakly forced so that its response is in the linear regime, $\delta^{(1)}[\Psi](t)$ can be observed from numerical simulations. In particular, if $g(t)=\delta(t)$, the perturbation would be equivalent to an $\varepsilon$-sized shift in the initial condition of the unperturbed system. This comment raises the question of using empirical measurements of the response $\delta^{(1)}\left[ \Psi\right](t)$ to find the decomposition of the Green function in terms of the point spectrum, which can be thence convoluted against different time-modulations. In effect, Eq.~\eqref{eq:eigenvalues linear response time dependent} provides a basis for the empirically observed response $\delta^{(1)}\left[ \Psi\right](t)$ to project onto.

In Fourier domain, the convolution product defining $\delta^{(1)}\left[\Psi \right](t)$ is transformed into a regular product:
\begin{equation}\label{eq:susceptibilty time dependent}
\mathfrak{F}\left[\delta^{(1)}\left[\Psi\right]\right](\omega)=\varepsilon \mathfrak{F}\left[g\ast \GGG\right](\omega) = \varepsilon \mathfrak{F}\left[g\right](\omega)\mathfrak{F}\left[\GGG\right](\omega),
\end{equation}
where $\mathfrak{F}$ indicates the application of the Fourier transform, $\omega$ is a real frequency and $\mathfrak{F}\left[\GGG\right]$ is the susceptibility function which is analytic in the upper complex plane, by virtue of the causality constraint. Furthermore, if one assumes polynomial decay of $\mathfrak{F}\left[\GGG\right](\omega)$, Kramers-Kronig relations \cite{Lucarini2005},  can be derived, which link the imaginary and real parts of the susceptibility function via Cauchy integrals \cite{ruellegeneral1998,lucarini2008}. Under suitable integrability assumptions, one can, moreover, link the logarithm of the modulus of the susceptibility with its real part, and then use Kramers-Kronig relations to find the imaginary counterpart \cite{Lucarini2012}.


%


Exploiting the spectral decomposition of quasi-compact semigroups and assuming that a given observable does not project onto the residual spectrum, the Fourier transform of the linear response is written as:

\begin{align}\label{eq:susceptibility transform}
\mathfrak{F}\left[\GGG \right](\omega) = \sum_{j=1}^{M}\sum_{k=1}^{a_j-1} \alpha_j^{(k)}\frac{1}{k!}\mathfrak{F} \left[  \Theta e^{\lambda_j \circ }\circ^k \right](\omega)=\sum_{j=1}^{M}\sum_{k=1}^{a_j-1} \frac{\alpha_j^{(k)}}{\lp i\omega - \lambda_j \rp^{k+1}}.
\end{align} 
The susceptibility function $\mathfrak{F}\left[\GGG\right]$ can be meromorphically extended to all values of $\omega$ in $\mathbb{C}$ such that $\mathfrak{Im}\omega > \omega _{ess}$, where the poles \cite{Ruelle1986} correspond to those of the resolvent norm $\| R(\cdot,\LLL_0) \|$ which are precisely given by the eigenvalues of $\LLL_0$ with the exception of that at $\omega = 0$ by the observation made below Eq.~\eqref{eq:alphas definition}. It is then clear that the location of the poles in the susceptibility function depends, uniquely, on the underlying system and not on the time-modulation function $g$, on the applied forcing, or on the observable in question. 

Correspondingly, Eq.~\eqref{eq:susceptibility transform} implies the existence of resonances in the response for $\omega=\mathfrak{Im} \lambda_j$, $j=1,\ldots,M$. Neglecting for the moment the possible existence of nonunitary algebraic multiplicities, the visibility of each resonance (the height of the corresponding peak when considering real values for $\omega$) will, instead, depend on the  width, determined by $\mathfrak{Re} \lambda_j$ (which depends neither on the observable  nor on the applied forcing), and by the weight given by the spectral coefficient $\alpha^{(1)}_j$ (which depends, instead, on both the observable and  the applied forcing). In the limit of a vanishing spectral gap $\gamma$, the spectrum will be dominated by the resonance associated to $\lambda_1$ unless $\alpha^{(1)}_j$ vanishes. 

\begin{remark}
The existence of a smooth invariant density for Eq.~\eqref{eq:sto ode 2} is enough for the fluctuation-dissipation theorem to hold \cite{pavliotisbook2014}. This implies that the spectral decomposition of correlation functions in stochastic systems, see \cite{chekroun2019c}, would equally yield Eqs.~\eqref{eq:eigenvalues linear response time dependent}~and~\eqref{eq:susceptibility transform}. However, in the limit of small noise, density representations and amenable functional settings for non-conservative can fail to exist. Attention is drawn, instead, to the Green function which is possible to define more generally using Ruelle's response theory for deterministic flows \cite{ruelle2009,lucarini2008}. The ultimate goal is therefore to replace $\rho_0(\xx)\dd \xx$ by $\rho_0(\dd\xx)$ in Eq.~\eqref{eq:linear response time dependent} and hope that the spectral decomposition of Eq.~\eqref{eq:eigenvalues linear response time dependent} holds. If such is the case, the relaxation rates would still be encoded in the Green function through the generator $\LLL^{\ast}_0$ while correlation functions would cease to explain the response at leading order.
\end{remark}

\section{Stochastically Perturbed Deterministic Systems}\label{sec:stochastically forced}

The theory of random dynamical systems regards stochastic flows as transformations of phase space parametrised in time but also in the noise realisation \cite{arnold1998}. In this sense, one can view noise as the inhomogeneous component of an equation that makes it, loosely speaking, non-autonomous. Interestingly, when the noise component is relatively weak, it can be regarded as a perturbation to the otherwise unperturbed deterministic system. Consequently, the leading order modification to the system's statistics is expected to be captured through response theory, in the spirit of the previous section.

Indeed, such is the case, although two main issues arise. First, the reference dynamics is deterministic, so that the matrix $\Sigma$ in Eq. \eqref{eq:sto ode 3} vanishes and hence, the unperturbed system will not, in general, posses an absolutely continuous invariant measure with respect to Lebesgue. This prevents the possibility of using density functions when studying the evolution of the measure. Second, the Green function can be seamlessly applied to derive realisation-dependent response formulas although these do not agree with a perturbation expansion of the Fokker-Planck equation. This section, therefore, aims at (i), determining the leading order correction to the statistics using operator relations and (ii), proving that Green function and operator based response formulas are equal only when the former is taken in the Stratonovich setting.



We proceed by considering a stochastically perturbed system that obeys the following It\^o SDE:
\begin{equation}\label{eq:sto ode}
\dd \xx (t) = \FF (\xx)\dd t + \varepsilon \Gamma(\xx)\dd W_t,
\end{equation}
where $\FF: \R^d \longrightarrow \R^d$ is the drift, $\Gamma: \R^d \longrightarrow \R^{d\times p}$ is a perturbation matrix, and $W_t$ denotes a $p$-dimensional Wiener process with correlation matrix equal to the identity. 
Because we are dealing with a deterministic system, we shall avoid density representations and use, instead, the backward-Kolmogorov equation with describes the evolution of expectation values of observables $\Psi$:

\begin{equation}\label{eq:stocastic perturbed kolmogorov}
\partial_t\Psi(\cdot,t)= \left( \LLL^{\ast}_0 + \varepsilon^2 \LLL^{\ast}_2  \right)\Psi(\cdot,t).
\end{equation}
Here, we are using the adjointness between the Fokker-Planck and the backward-Kolmogorov equation (also invoked in Eq.~\eqref{eq:linear response time dependent}), whereby the operators $\LLL_0^{\ast}$ and $\LLL_2^{\ast}$ are given by:
\begin{subequations}
\begin{align}
	\LLL^{\ast}_{0}\Psi  &= \FF \cdot \nabla \Psi \label{L10ast}; \\
	\LLL^{\ast}_{2}\Psi  &= \frac{1}{2}\lp\Gamma\Gamma^{\top} \rp : \nabla^{2}\Psi, \label{eq:pert op 2}
\end{align}
\end{subequations}
for $\Psi $ in $D(\LLL^{\ast}_{2})$. Here we immediately observe that stochastic forcing enters at second order in $\varepsilon$ when considering the operator representation of the stochastic process \eqref{eq:sto ode}. 
Using the semigroup expansions of Section~\ref{sec:time dependent forcing}, we can derive the leading order response of the steady state to the introduction of noise:
	\begin{align}\label{eq:response sto perturbations}
		\delta^{(2)}\left[\Psi\right](t) &=\varepsilon^2 \int \rho_0(\dd \xx) \int_{0}^t\LLL_{2}^{\ast}\Psi(\xx(t))\dd s =\frac{\varepsilon^2}{2}\int\rho_0(\dd \xx) \int_{0}^t\lp\Gamma(\xx)
		\Gamma^{\top}(\xx)  \rp : \nabla^{2}\Psi(\xx(s)))\dd s.
	\end{align}
Note that we have evaluated the deterministic Koopman operator $e^{t\LLL_0^{\ast}}$ to the observable $\Psi$ to give $e^{t\LLL_0^{\ast}}\Psi(\xx)=\Psi(\xx(t))$, where $\xx(t)$ solves Eq.~\eqref{eq:sto ode} for $\varepsilon = 0$.

Formula Eq.~\eqref{eq:response sto perturbations} coincides with what was obtained in \cite{Abramov2017}, where instead, the author resorted to computing the impact of th eperturbation directly on the tangent space in order to evaluate $\lp\Gamma(\xx)
\Gamma^{\top}(\xx)  \rp : \nabla^{2}\Psi(\xx(t))$. 
An approach based on second-order response theory was followed by \cite{Lucarini2012}, yet yielding a different formula compared to Abramov's \cite{Abramov2017}. In such a work, the Wiener increments $\dd W_t$ were treated as the time-modulation of a sequence of $p$ perturbations with spatial patterns defined by the columns of the $\Gamma$ matrix. To illustrate this idea, it is useful to rewrite Eq.~\eqref{eq:sto ode} as a series of applied vector fields modulated by Wiener increments:
\begin{align}\label{decomposed stochastic perturbation}
	\dd \xx = \FF (\xx) \dd t + \varepsilon \sum_{i=1}^p g_i(t)\Gamma_{:,i} (\xx),
\end{align}
where $\Gamma_{:,i}$ is the $i$th column of $\Gamma$ and $g_i(t)$ are assumed to satisfy $\E \left[g_i(t)\right]=0$ and $\E \left[g_k(t)g_l(s)\right]=\delta(t-s)\delta_{k,l}$ 
for every $i,k,l=1,\ldots,p$, where $\delta_{k,l}$ is the Kronecker delta. For a given realisation of the noise indexed by $\sigma$, one can apply Eq.~\eqref{eq:linear response time dependent} to obtain the linear response correction:
\begin{align}\label{eq:linear response before averaging}
\delta ^{(1)}_{\sigma}\left[ \Psi \right] (t) := \varepsilon  \sum_{i=1}^{p}\int_{-\infty}^{\infty} \GGG_{i}(s)g_i(t-s)\dd s,
\end{align}
where $\GGG_{i}(t) = \Theta(t)\int \rho_0(\dd \xx) \LLL^{\ast}_1e^{t\LLL^{\ast}_0}\Psi(\xx) $ is the Green function associated with the vector field $\Gamma_{:,i}$ for every $i$. Since, $g_i$ corresponds to a Wiener increment, taking averages over all realisations $\sigma$ makes the linear responses vanish: $\mathbb{E}\left[  \delta ^{(1)}_{\sigma}\left[ \Psi \right] \right]= 0$. Contrarily, the second order response $\delta^{(2)}_{\sigma}\left[\Psi \right] $ survives the averaging and is simplified to a single-time-variable integral of the form:
	\begin{align}
		\tilde{\delta}^{(2)}\left[\Psi  \right](t) &:= \mathbb{E}\left[  \delta ^{(2)}_{\sigma}\left[ \Psi \right] \right]=  \frac{ \varepsilon^2 }{2} \sum_{k=1}^p\int \rho_0(\dd\xx)\int_{-\infty}^{\infty} \Theta (s) \LLL_{1,k}^{\ast}\LLL_{1,k}^{\ast}e^{s\LLL_0^{\ast}}\Psi(\xx)\dd s, \label{eq:lucarini formula 1} 
	\end{align}
where $\LLL^{\ast}_{1,i} = \Gamma_{:,i}\cdot \nabla$. To compare Eq.~\eqref{eq:lucarini formula 1} with Eq.~\eqref{eq:response sto perturbations}, it is handy to disentangle the column vector fields $\Gamma_{:,i}$ from the correlation matrix $\Gamma$ so that, after some vector-matrix manipulations, one deduces, instead, that:
\begin{align}\label{eq:second order response operator form}
	\delta^{(2)}\left[\Psi\right](t)= \frac{ \varepsilon^2 }{2}\sum_{k=1}^p
	\int \rho_0(\dd \xx) \int_0^t\lp \Gamma_{:,k}(\xx)\Gamma_{:,k}^{\top}(\xx) \rp : \nabla ^2\Psi(\xx(s))\dd s.
\end{align}

\begin{remark}
	In case the applied fields $\Gamma_{;,i}$ in Eq.~\eqref{decomposed stochastic perturbation} are not modulated by independent time-functions $g_i(t)$, we would have that $\mathbb{E}\left[g_k(t)g_l(s)\right]=\delta(t-s)c_{k,l}$, for some correlation factor $c_{k,l}$. This would make cross-terms appear, so that
	\begin{align}
		\tilde{\delta}^{(2)}\left[\Psi  \right](t) &:= \mathbb{E}\left[  \delta ^{(2)}_{\sigma}\left[ \Psi \right] \right]= \frac{ \varepsilon^2 }{2} \sum_{k=1}^p\sum_{l=1}^p C_{k,l} \int \rho_0(\dd\xx)\int_{-\infty}^{\infty} \Theta (s) \LLL_{1,k}^{\ast}\LLL_{1,l}^{\ast}e^{s\LLL_0^{\ast}}\Psi(\xx)\dd s,  
	\end{align}
and,
\begin{align}
	\delta^{(2)}\left[\Psi\right](t)= \frac{ \varepsilon^2 }{2}\sum_{k=1}^p\sum_{l=1}^p C_{k,l}
	\int \rho_0(\dd \xx) \int_0^t\lp \Gamma_{:,k}(\xx)\Gamma_{:,l}^{\top}(\xx) \rp : \nabla ^2\Psi(\xx(x))\dd s.
\end{align}
However, it is easy to recast this case into the form of Eq.~\eqref{eq:sto ode} upon a modification of the covariance matrix.
	
\end{remark}

We remark that the approach described in the present paper using operator expansions yields a different formula to that of \cite{Lucarini2012} for the seemingly same question. In order to understand this discrepancy, we need to clarify our interpretation of the noise term. The algebraic manipulations used in \cite{Lucarini2012} rely on using the standard rules of calculus, which is appropriate if one considers Stratonovich's interpretation of noise. And, indeed, the concept behind the use of the second order response approach relied on treating the noise terms as fast forcings, thus considering the physical limit of time-scale separation going to infinity. This, again, points in the direction of Stratonovich's interpretation \cite{Pavliotis2008,pavliotisbook2014}.

To test our hypothesis,  we shall investigate whether the approach taken in \cite{Lucarini2012} can be recast more rigorously according to the formalism used in this paper. To this end and for clarity, let us consider Eq.~\eqref{decomposed stochastic perturbation} where a single perturbation $\Gamma_{:,1}$ is applied, namely, $p=1$. This means 
we can rewrite Eq.~\eqref{decomposed stochastic perturbation} as:
\begin{equation}\label{eq:sto ode rewritten}
\dd \xx (t) = \FF (\xx)\dd t + \varepsilon \tilde \Gamma H \dd W_t,
\end{equation}
where $\xx (t)$ is in $\R^d$, $\tilde \Gamma =  \mathrm{diag}(\Gamma_{:,1})$, $W_t$ is a $d$-dimensional Wiener process and $H$ is a $d\times d$ matrix defined as:
\begin{equation}
H=\mathbbm{1}_d\mathbf{e}_{1}^{\top}=\begin{bmatrix}
1 \\ \vdots \\ 1
\end{bmatrix}[1,0,\ldots,0].
\end{equation}
We proceed by expanding the operator product $\LLL_{1,1} ^{\ast}\LLL_{1,1}^{\ast}$, which is instrumental in  Eq.~\eqref{eq:lucarini formula 1}. Letting $f$ denote a generic twice differentiable function:
\begin{subequations}
	\begin{align}
		\LLL_{1,1}^{\ast}&\LLL_{1,1}^{\ast}f = \Gamma_{:,1} \cdot \nabla \lp \Gamma_{:,1} \cdot \nabla f \rp \\&= \Gamma_{:,1} \cdot \left[ \lp \Gamma_{:,1} \cdot \nabla \rp \nabla f + \lp \nabla f \cdot \nabla \rp \Gamma_{:,1}\right] \nonumber \\&\hspace{1.2cm}+ \Gamma_{:,1} \cdot\left[\Gamma_{:,1} \times (\nabla \times \nabla f)+\nabla f \times (\nabla \times \Gamma_{:,1})  \right] \\&= \Gamma_{:,1} \cdot \left[ \lp \Gamma_{:,1} \cdot \nabla \rp \nabla f + \lp \nabla f \cdot \nabla \rp \Gamma_{:,1} +\nabla f \times (\nabla \times \Gamma_{:,1})  \right] \\&= \Gamma_{:,1} \cdot \lp \Gamma_{:,1} \cdot \nabla \rp \nabla f +\Gamma_{:,1}\cdot  \lp \nabla f \cdot \nabla \rp \Gamma_{:,1} +\Gamma_{:,1} \cdot \lp \nabla f \times \nabla \times \Gamma_{:,1} \rp\\&= \Gamma_{:,1} \cdot \lp \Gamma_{:,1} \cdot \nabla \rp \nabla f +\Gamma_{:,1}\cdot  \lp \nabla f \cdot \nabla \rp \Gamma_{:,1} -\nabla f \cdot \lp \Gamma_{:,1} \times \nabla \times \Gamma_{:,1} \rp\\&= \lp \Gamma_{:,1} \Gamma_{:,1}^{\top} \rp :\nabla^2f +\Gamma_{:,1}\cdot  \lp \nabla f \cdot \nabla \rp \Gamma_{:,1} -\nabla f \cdot \lp \Gamma_{:,1} \times \nabla \times \Gamma_{:,1} \rp \label{eq: deriv 7}.
	\end{align}
\end{subequations}
We immediately observe that if $\nabla \Gamma_{:,1} \equiv 0$, the response formula derived in Eqs.~\eqref{eq:response sto perturbations} and \eqref{eq:second order response operator form} coincide. Consequently, agreement is found if the forcing is additive white noise. The second and third terms on the RHS of Eq.~\eqref{eq: deriv 7} are extra terms that can be associated with the It\^o-to-Stratonovich correction \cite{pavliotisbook2014}. Indeed, we see this by regarding Eq.~\eqref{eq:sto ode rewritten} in the Stratonovich sense:
\begin{equation}\label{eq:sto ode rewritten stratonovich}
\dd \xx (t) = \FF (\xx)\dd t + \varepsilon \tilde \Gamma H \circ \dd W_t.
\end{equation}
The corresponding It\^o conversion reads and expands as
\begin{subequations}\label{eq:sto ode converted}
	\begin{align}
		\dd \xx (t) &= \left[\FF (\xx) +  \frac{\varepsilon^2}{2}\left[ \nabla \cdot \lp (\tilde \Gamma H)(\tilde \Gamma H)^{\top} \rp - \lp \tilde \Gamma H \rp \nabla \cdot \lp \tilde \Gamma H\rp \right] \right]\dd t + \varepsilon \tilde \Gamma H \dd W_t \\&=\left[\FF (\xx) +  \frac{\varepsilon^2}{2}\left[ \nabla \cdot \lp \Gamma_{:,1}\Gamma_{:,1}^{\top}\rp - \lp\nabla \cdot \Gamma_{:,1} \rp\Gamma_{:,1}  \right] \right]\dd t + \varepsilon \tilde \Gamma H  \dd W_t .
	\end{align}
\end{subequations}
Note that an $\varepsilon^2$ term has appeared in the drift component. Consequently, the backward-Kolmogorov equation associated with Eq.~\eqref{eq:sto ode converted} will only have perturbation operators of order $\varepsilon ^2$. Following the expansions that lead to Eq.~\eqref{eq:response sto perturbations}, the perturbation operator $\LLL^{\ast}_{2}$ for Eq.~\eqref{eq:sto ode converted} writes as:
\begin{subequations}
	\begin{align}
		\LLL^{\ast}_{2}f &= \frac{1}{2}\left[ \nabla \cdot \lp \Gamma_{:,1}\Gamma_{:,1}^{\top}\rp - \lp\nabla \cdot \Gamma_{:,1} \rp\Gamma_{:,1}  \right]\cdot \nabla f +\frac{1}{2}\Gamma_{:,1}\Gamma_{:,1}^{\top}:\nabla^2f  \\&=\frac{1}{2}\left[ \lp \Gamma_{:,1} \cdot \nabla \rp \Gamma_{:,1}  \right]\cdot \nabla f + \frac{1}{2}\Gamma_{:,1}\Gamma_{:,1}^{\top}:\nabla^2f \\&=\frac{1}{2}\left[ \frac{1}{2}\nabla \lp \Gamma_{:,1} \cdot \Gamma_{:,1} \rp - \Gamma_{:,1} \times \nabla \times \Gamma_{:,1} \right]\cdot \nabla f + \frac{1}{2}\Gamma_{:,1}\Gamma_{:,1}^{\top}:\nabla^2f. \label{eq:expansion second order stratonovich}
	\end{align}
\end{subequations}
Hence, comparing Eq.~\eqref{eq:expansion second order stratonovich} and Eq.~\eqref{eq: deriv 7}, we are left with showing that $\nabla \lp \Gamma _{:,1}\cdot \Gamma_{:,1} \rp \cdot \nabla f = 2\Gamma_{:,1}\cdot \lp \nabla f \cdot \nabla  \rp \Gamma_{:,1} $, which can be seen by direct evaluation:
\begin{subequations}
	\begin{align}
		\frac{1}{4}\nabla \lp \Gamma_{:,1} \cdot \Gamma_{:,1} \rp \cdot \nabla f &= \frac{1}{4}\begin{bmatrix}
			2\Gamma_{1,1}\partial_{x_1}\Gamma_{1,1} + \ldots +2\Gamma_{d,1}\partial_{x_1}\Gamma_{d,1} \\ \vdots \\2\Gamma_{1,1}\partial_{x_d}\Gamma_{1,1} + \ldots +2\Gamma_{d,1}\partial_{x_d}\Gamma_{d,1}
		\end{bmatrix}\cdot \nabla f \\&= \frac{1}{2}\left[ \Gamma_{:,1} \cdot \partial_{x_1}f\partial_{x_1}\Gamma_{:,1} + \ldots + \Gamma_{:,1} \cdot \partial_{x_d}f\partial_{x_d}\Gamma_{:,1} \right] \\&=\frac{1}{2}\Gamma_{:,1} \cdot \lp \nabla f \cdot \nabla \rp\Gamma_{:,1}.
	\end{align}
\end{subequations}
Which proves the claim. The general case is formalised in the following proposition:
\begin{proposition}\label{prop:ito stratonovich}
	Consider the SDE~\eqref{eq:sto ode} in Stratonovich form and let $\mathbf{H}:\R^d\longrightarrow \R^d$ be the  Stratonovich-to-It\^o correction:
	\begin{equation}\label{eq:ito strat correction propostition}
	\mathbf{H} =  \frac{1}{2}\left[ \nabla \cdot \lp \Gamma  \Gamma^{\top} \rp - \Gamma \lp  \nabla \cdot \Gamma^{\top}  \rp  \right].
	\end{equation}
	Then,
	\begin{equation}
	\tilde{\delta}^{(2)}\left[\Psi \right] - \delta^{(2)}\left[\Psi \right] = 
	\varepsilon^2\int\rho_0(\dd\xx) \int_{-\infty}^{\infty} \Theta(s)\HH(\xx)\cdot\nabla\Psi(\xx(s))\dd s,
	\end{equation}	
	where the first and second terms of the left-hand side are defined in Eq.~\eqref{eq:lucarini formula 1} and Eq.~\eqref{eq:second order response operator form}, respectively.
\end{proposition}
\begin{proof}
	We note the following chain of equalities:
	\begin{subequations}
		\begin{align}
			\dd \xx &= \FF \dd t + \varepsilon \Gamma  \circ \dd W_t = \left[\FF + \HH \right]\dd t + \varepsilon \Gamma  \dd W_t\\&=\left[\FF + \varepsilon^2\sum_{k=1}^{p} \HH_{k} \right]\dd t + \varepsilon\sum_{i=1}^p\tilde{\Gamma}_iH_i \dd W_t,
		\end{align}
	\end{subequations}
	where $\tilde{\Gamma}_i = \mathrm{diag}\lp \Gamma_{:,i} \rp$, $H_i=\mathbbm{1}_d\mathbf{e}_i^{\top}$ and $\HH_{k}$ is defined as:
	\begin{equation}
	\HH_{k}=\frac{1}{2}\left[ \nabla \cdot \lp \Gamma_{:,k}\Gamma_{:,k}^{\top}\rp - \tilde \Gamma_{k}\nabla \cdot \tilde \Gamma_k\right].
	\end{equation}
	Now, it is enough to repeat the argument started in Eq.~\eqref{eq:sto ode rewritten} for each $\tilde{\Gamma}_i$ and then use the linearity of the leading order response to extend it for $i=1,\ldots,p$.\end{proof}

\section{Response of Correlations and Power Spectra}\label{sec:response correlations}
In \cite{lucarini2017b} it was shown how to compute the change in the correlation properties of a deterministic chaotic system resulting from a static forcing applied to its dynamics. This amounts to considering the case described in Eq. 
\eqref{eq:sto ode 2} with $\Sigma=0$ and with $g(t)=1$.
 Indicating with  $C^{\varepsilon}_{\Psi,\Phi}$  the  correlation function as  two square-integrable observables $\Psi$ and $\Phi$ in the perturbed dynamics, 
 one gets:
\begin{subequations}\label{eq:perturbed correlations}
	\begin{align}
	C^{\varepsilon}_{\Psi,\Phi}(t)=&\int \rho_\varepsilon(\dd \xx)e^{t(\LLL_0^{\ast}+\varepsilon\LLL_1^\ast)}\Psi(\xx) \Phi(\xx)\\
	&=\int \rho_0(\dd \xx)e^{t\LLL_0^{\ast}}\Psi(\xx) \Phi(\xx) \label{eq:referencecorrelations}  \\&+\varepsilon\int \rho_0(\dd \xx)\int_{0}^{\infty} \LLL_1^\ast e^{(t+s)\LLL_0^{\ast}}\Psi(\xx) e^{s\LLL_0^{\ast}}\Phi(\xx) \dd s \label{eq:perturbed correlations 1} \\&+\varepsilon\int \rho_0(\dd \xx) \int_{0}^{t} e^{(t-s)\LLL_0^{\ast}}\LLL_1^{\ast}e^{s\LLL_0^{\ast}}\Psi(\xx)\Phi(\xx) \dd s + \mathcal{O}\lp \varepsilon^2 \rp.\label{eq:perturbed correlations 2}
	\end{align}
\end{subequations}
where $\LLL_0^\ast$ and $\LLL_1^\ast$ are given in Eq.~\eqref{L10ast} and \eqref{L1asta}, respectively. The term shown in Eq.~\eqref{eq:referencecorrelations} gives the unperturbed correlation. The leading order correction to correlation functions consists of two different components corresponding to the perturbation expansion of the invariant measure and the Koopman operator, respectively. 
This term can also be written  as 
\begin{equation}
  \varepsilon\int \rho_1(\dd \xx)\int_{0}^{\infty}  e^{t\LLL_0^{\ast}}\Psi(\xx) \Phi(\xx),\label{rho1modcorr} 
\end{equation}
which shows that it can be interpreted as the expectation value of the unperturbed form of time-delayed product of the two observables in the first order correction of the measure. 
Instead, the contribution given in Eq.~\eqref{eq:perturbed correlations 2} is qualitatively different and vanishes if $t=0$, because in such a case the regular response theory for observables applies. 

We  can generalise the previous result in two different directions. If we stick to deterministic background dynamics (so that $\LLL_0^\ast$ is given by Eq. \eqref{L10ast}) but consider, instead, adding stochastic forcing to the system--- as in Eq.~\eqref{eq:sto ode}--- the change in the correlation properties of the system can be evaluated by performing the following substitutions in Eqs.~\eqref{eq:perturbed correlations}-\eqref{rho1modcorr}: $\LLL_1^\ast\rightarrow \LLL_2^\ast$, where the latter operator is defined in Eq.~\eqref{eq:pert op 2}; $\rho_1\rightarrow \rho_2$, where the second-order correction to the invariant measure is defined in Eq.~\eqref{eq:expansion time dependent density}; and $\varepsilon\rightarrow\varepsilon^2$. If, instead, we consider a stochastic reference dynamics (so that $\LLL_0^\ast$ is given by Eq.~\eqref{L0aststo}) and want to study the impact of changing the noise law, we can insert the operator $\LLL_1^\ast$ given in Eq.~\eqref{L1astb} into Eqs.~\eqref{eq:perturbed correlations}. Of course Eqs.~\eqref{eq:perturbed correlations} apply seamlessly in the case we alter the drift term of the stochastic differential equation.

Applying the Fourier transform to $C^{\varepsilon}_{\Psi,\Phi}(t)$  yields the response in spectral domain, again following \cite{lucarini2017b}. We can, however, use the resolvent formalism and Laplace transforms of semigroups to express the correlation function in spectral domain \cite{engel2000}. The Laplace transform of a correlation function reads as:
\begin{equation}
\mathfrak{L}[C^{\varepsilon}_{\Psi,\Phi}](z)= \int \rho_{\varepsilon}(\dd \xx) R(z,\LLL^{\ast}_0+\varepsilon \LLL^{\ast}_1)\Psi(\xx) \Phi(\xx) ,
\end{equation}
for $\mathfrak{Re}z>0$ and where $R (z,\LLL_0) = \left( z - \LLL_0 \right)^{-1}$ is the resolvent operator. Expanding $\mathfrak{L}[C^{\varepsilon}_{\Psi,\Phi}](z)$ in powers of $\varepsilon$ we get a linear response for correlation functions in terms of derivatives with respect to $\varepsilon$:
\begin{equation}\label{eq:laplace 1}
\begin{aligned}
\frac{\dd }{\dd \varepsilon}\mathfrak{L}[C^{\varepsilon}_{\Psi,\Phi}](z)|_{\varepsilon=0} =  \int \rho_1(\dd \xx) R(z , \LLL^{\ast}_0 )\Psi(\xx)\Phi(\xx) + \int \rho_0(\dd \xx) R(z,\LLL^{\ast}_0)\LLL^{\ast}_1R(z,\LLL^{\ast}_0)\Psi(\xx)\Phi(\xx).
\end{aligned}
\end{equation}
This formula can be obtained by taking the Laplace transform of Eq.~\eqref{eq:perturbed correlations}, although the resolvent formalism makes the calculations substantially easier. We, thus, provide an explicit link between the response in frequency domain and the spectral properties of the unperturbed generator $\LLL^{\ast}_0$ via the resolvent. Similarly to the discussion above, upon the replacement of $\LLL_0^\ast$ or $\LLL_1^{\ast}$, one can deduce the response formula for deterministic or stochastic perturbations.




In \cite{Lucarini2012} linear response of power spectral densities is related to the modulus of the susceptibility function in case of white-noise modulated forcings. Indeed, there it is argued that if $\mathfrak{F}\left[C^{\varepsilon}_{\Psi,\Psi} \right](\omega)$ is the perturbed power spectrum of the observable $\Psi$ then we have:
\begin{equation}\label{eq:kk susceptibility}
\mathbb{E}\left[\mathfrak{F}\left[C^{\varepsilon}_{\Psi,\Psi} \right](\omega)\right]-\mathfrak{F}\left[C_{\Psi,\Psi} \right](\omega)\approx \mathbb{E}\left[  \left| \delta ^{(2)}_{\sigma}\left[ \Psi \right]\right|^2 \right]\approx \varepsilon^2 \left| \mathfrak{F}\left[\GGG\right](\omega) \right|^2.
\end{equation}
In order to prove this, the author of \cite{Lucarini2012} invoked the Wiener-Khinchin theorem and applied it to the autocorrelation function of $\delta^{(1)}_{\sigma}[\Psi](t)$ before taking averages over $\sigma$. Equation~\eqref{eq:kk susceptibility} appears to have a term missing if one compares it to Eq.~\eqref{eq:laplace 1}. This is because that the author uses $\mathfrak{F} \left[ C^\varepsilon_{\Psi,\Psi} \right]\approx \mathfrak{F} \left[ \int \rho_0(\dd \xx) \Psi(\xx)e^{t\LLL}\Psi(\xx) \right]$ as an approximation. Indeed, before averaging over realisations, the leading order response of the flow is proportional to $\varepsilon$ whereas the perturbed invariant measure responds proportionally to $\varepsilon^2$. Therefore, the RHS of Eq.~\eqref{eq:laplace 1} yields a single term from which Eq.~\eqref{eq:kk susceptibility} follows.

Notice that the left hand side of the equation above is especially easy to calculate if one is let to sample the perturbed and unperturbed dynamics. Thus, it is possible to obtain estimates of the modulus of the susceptibility function in a fully empirical way. One is left with finding the actual locations of the poles, which the modulus itself cannot reveal. This would be attained by noticing that the logarithm of the susceptibility function can be written as:
\begin{equation}\label{eq:logarithm of susceptibility}
\log \lp \mathfrak{F}\left[ \GGG \right](\omega) \rp= \log \lp\left| \mathfrak{F}\left[ \GGG \right](\omega) \right|\rp + i \varphi(\omega),
\end{equation}
for every $\omega$ and some phase function $\varphi$. Assuming, further, that $\log \lp \mathfrak{F}\left[ \GGG \right](\omega) \rp $ is analytic, Kramers-Kronig relations are invoked to link its real and imaginary parts shown in Eq.~\eqref{eq:logarithm of susceptibility} and, thus, that of $\mathfrak{F}\left[ \GGG \right](\omega) $. This way, Eq.~\eqref{eq:kk susceptibility} constitutes a practical way of estimating the susceptibility function of a deterministic flow by using time series of the stochastically forced vector field.

\section{Conclusions}\label{sec:general remarks}

In this note we have shown how the formalism of operator semigroups allows for a compact derivation of response formulas describing at leading order the impact of adding forcings of either deterministic or stochastic nature in both cases of systems whose background dynamics is deterministic or stochastic. Linear response  gives the leading order correction to expectation values due to extra deterministic forcings acting on either stochastic or chaotic dynamical systems. The response is also linear in the intensity of the (extra) noise for background stochastic dynamics. Note that we are able to accommodate results pertaining to non-stationary noise laws. 

In the case of stochastic systems, the linear response formulas have been expressed through a spectral decomposition on the basis of the eigenvectors and eigenfunctions of the unperturbed Fokker-Planck/backward-Kolmogorov generator; see Eq.~\eqref{eq:eigenvalues linear response time dependent}. This angle gives a basis for bottom-up construction of response functions and relaxation rates for non-equilibrium systems, and clarifies the general applicability and interpretation of the FDT. It also casts the computation of the response of an observable to perturbations in the context of the problem of finding suitable and accurate estimates for the modes of the Koopman operator \cite{Mezic2005,Budinisic2012}. This research area has benefitted considerably from the use of data-driven methods \cite{Kutz2016,Lusch2018}. Koopman modes are ordered according to their decay rate, thus giving a rationale for, in most cases, the greater relevance of the contributions coming from slowly decaying modes; see Eqs. \eqref{eq:eigenvalues linear response time dependent}-\eqref{eq:alphas definition}. {\color{black}We remark that such a connection between response formulas and Koopmanism applies under the conditions that the residual spectrum decays on considerably faster time-scales than the point spectrum and/or that we consider observables whose projection on the modes associated with such residual spectrum is negligible. An instance where this framework holds can be found in \cite[Section D]{Wouters2016,SantosGutierrez2021}. Broadly, this indicates that our approach could be suitable for the cases where model reduction techniques based on the Koopman operator are applicable \cite{Budinisic2012,Tu.ea.2014,Lusch2018}.} 

Future endeavours should aim at finding computationally amenable ways of identifying the eigenvalues and the associated timescales  appearing in Eq.~\eqref{eq:eigenvalues linear response time dependent}. We remark that this problem is attracting considerable interest in the climate community, because it would allow to draw a convincing link between climate variability and climate response \cite{Ghil2020}. Nonetheless, recent proposals in this direction--- while indeed interesting and promising \cite{Torres2021a,Torres2021b,Bastiaansen2021}--- are yet unable to deal with complex eigenvalues, thus missing out an essential ingredient of the dynamics of nonequilirium systems \cite{SantosGutierrez2021}. Empirically-learnt Markov-matrices are a promising technique that should be explored, even taking into account the need to necessarily operate in a much reduced dimension phase space for any application where the dynamics is high-dimensional \cite{chekrounrough2014,chekroun2019c}. Note that in the case of Markov chains it is easy to estimate the actual radius of convergence of the response operators, i.e. the maximum intensity of the applied forcing \cite{Lucarini2016,SantosJSP}. Contrarily, we highlight the results of \cite{kenkre1971, kenkre1973} which present a non-perturbative response formula based on the Mori-Zwanzig projection operator \cite{zwanzig}; these formula is revisited in Appendix~\ref{sec:homogeneous equation}. In this paper we have on purpose avoided to discuss any aspect of convergence of the perturbative formulas presented above; this is left for future investigations.


When stochastically perturbing systems whose underlying dynamics is chaotic, the second order-response gives the leading-order correction. We have been able to reconcile two results published in the literature and showed that the formulas given in \cite{Lucarini2012} and in \cite{Abramov2017} coincide when one interprets the results of the former study taking into account the physically relevant Stratonovich interpretation of noise, which indeed complies with the heuristic procedure followed there, where the standard rules of calculus had been used throughout the paper. This result helps to better understand how the invariant measure convergences to the physical one when the intensity of the noise becomes infinitesimal. 



Finally, we have adapted the mathematical technology developed for studying the change in the expectation value of observables due to the application of external forcings combined with the  perturbative expansion of the resolvent operator \cite{engel2000} to generalise the results presented in \cite{lucarini2017b}. We have developed leading response formulas for the two-point correlations of either deterministic or stochastic dynamical systems. Also in this case, it is easy to show that the order of correction is linear in the case of deterministic  perturbation to either stochastic or chaotic systems and of stochastic perturbations of stochastic systems, while it is quadratic in the case of purely stochastic perturbations applied to chaotic systems. The weak noise assumption is fundamental to derive Eq.~\eqref{eq:kk susceptibility}, originally presented and discussed with the help of numerical simulations in \cite{Lucarini2012}. A more detailed study of this relation would be most valuable since its implementation is straightforward and could provide new operational algorithms for computing the spectrum of the Fokker-Planck operators and the subsequent susceptibility functions, which are key to understand critical transitions in complex systems \cite{Lucarini2018JSP,Tantet2018,Tantet2018b}. 

On a more general note, the effects of stochastic noise on the forced variability of a deterministic system are yet to be fully understood. For instance, numerical studies show that under the weak-noise assumption--- potentially at the linear response level, as in Section~\ref{sec:stochastically forced}--- the low-frequency variability spectrum of a multi-modal deterministic system can be heavily amplified as results of an apparent stabilisation of the different coexisting regimes \cite{kwasniok2014}. While the counterintuitive fact that noise stabilise an otherwise unstable deterministic system is well known, it is clear that the underlying dynamical structure of the deterministic system is also essential for such behaviour \cite{Shen2007,Turcotte2008,Parker2011,Clusella2017,Meyer2018,Kolba2019,Manchein2022}. Indeed, it has been argued that the systems displaying intermittency or activation are more likely to be stabilised under the noisy perturbations \cite{dorrington_thesis}. 
The present paper does not tackle this question, although we here provide formal methods to treat systems subject to general external forcings and it is hoped that our results will help understand the interplay between deterministic and stochastic feature of the dynamics of rather general class of systems featuring non-trivial variability.

\paragraph{Acknowledgements.}

The authors would like to thank Micka\"el Chekorun, Jeroen Wouters, and Joshua Dorrington for their encouragement, comments, and suggestions. VL acknowledges the support received from the Horizon 2020 project TiPES (grant no. 820970) and from the EPSRC project EP/T018178/1.

\appendix

\section{Kenkre's Response Formula}\label{sec:homogeneous equation}

The limit of weak forcing strength is instrumental in deriving meaningful linear response formulas that capture statistical changes due to an applied external field. When small, yet finite-sized perturbations are introduced, power series are resorted to and their convergence depends on a radius of expansion. To obtain the full response to any order in $\varepsilon$, one seems inevitably obliged to solve the perturbed Fokker-Planck equation~\eqref{eq:fpe 2} to find the probability state at a certain time to evaluate the phase average a given observable. In the present section, we aim at deriving a scalar equation for the response function $\langle \Psi , \rho(\cdot,t) \rangle$ in Eq.~\eqref{eq:response function a} that does not require a power expansion nor the full knowledge of perturbed probability density function. To do so, we revisit the calculations presented in \cite{kenkre1971, kenkre1973} in the context of quantum mechanical systems, at a time when linear response and the subsequent fluctuation-dissipation theorem were just being introduced. Now, well after Ruelle's mathematical proof of linear response, it is our desire to review such results but in a stochastic framework and with references to recent advances in the treatment of operator equations. To do this, the equation of motion--- regardless it being stochastic or deterministic--- will suitably be projected to derive, with the aid of the Mori-Zwanzig formalism \cite{zwanzig}, a scalar equation for the response $\langle \Psi , \rho(\cdot,t) \rangle$ where the only independent variable is $\langle \Psi , \rho(\cdot,t) \rangle$ itself; this is, an \emph{homogeneous} equation.


To achieve this goal, we shall start by considering the forced SDE~\eqref{eq:sto ode 2}. Since the introduction of perturbations makes the equations of motion non-autonomous, the usual exponential operator no longer serves as a notation for the solution of Eq.~\eqref{eq:fpe 2} unless we consider \emph{time-ordered} exponentials \cite{gill2017}.
It is useful to define the full response function $R_{\Psi}(t)$ as:
\begin{equation}\label{eq:full response function}
R_{\Psi}(t):=\left \langle \Psi ,\rho(\cdot,t) \right \rangle = \int \Psi (\xx)\rho(\xx,t)\dd\xx,
\end{equation}
where $\rho(\cdot,t)$ is a solution of Eq.~\eqref{eq:fpe 2}. Traditionally, the response function is defined as the difference between the perturbed and unperturbed means of the observable $\Psi$, but since the latter is stationary, this definition will not introduce differences with other formulations.

The idea is now to project the Fokker-Planck equation according to $\mathbb{P}f = \left(\int\Psi f \right)\rho_0 $, so that the Mori-Zwanzig formalism yields an equation in terms of projected operators instead of the full ones. Upon substitution and simplification we find that the evolution equation for the response function $R_{\Psi}$ can be written in more compact terms as:
\begin{align}\label{eq: full response equation}
\partial _t R_{\Psi}(t) &=B(t)R_{\Psi}(t)+\int_{0}^t K(t,s) R_{\Psi}(s)\dd s,
\end{align}
where the newly introduced function $B(t)$ and memory kernel $K(t,s)$ are defined as:
\begin{subequations}
	\begin{align}
	B(t)&=\xi^{-1}\int\Psi(\xx)\mathcal{L}(t)\rho_0(\xx)\dd \xx \label{eq: B function general}, \\ \label{eq: kernel general}
	K(t,s)&=\xi^{-1} \int\Psi(\xx)\mathcal{L}(t)e^{\int _{s}^t\lp 1 - \mathbb{P} \rp \mathcal{L}(\tau)\dd \tau}\lp 1 - \mathbb{P} \rp \mathcal{L}(s)\rho_0(\xx)\dd \xx,
	\end{align}
\end{subequations}
where $\xi = \left\langle \Psi , \rho_0 \right\rangle \neq 0$. Equation~\eqref{eq: full response equation} becomes a scalar and homogeneous equation on $R_{\Psi}$, where there is not explicit reference to the time-dependent measure $\rho(\xx, t)$ that would be obtained by solving the full Fokker-Planck equation \eqref{eq:fpe 2}. Collecting the leading order terms in Eq.~\eqref{eq: full response equation} one obtains the usual linear response, as shown in \cite{kenkre1971}. 

Solving the non-Markovian equation \eqref{eq: full response equation} would yield the response function and, in principle, should be equal to the perturbative approach in Eq.~\eqref{eq:response function a}, provided that the parameter $\varepsilon$ is within the radius of convergence. In the work of \cite{kenkre1971}, there is no proof of the convergence of Eq.~\eqref{eq: full response equation}, which would boil down to determining the decay of $K(t)$ as $t$ tends to infinity. If one assumes that $g$ is constant and that $\LLL$ generates a quasi-compact semigroup--- see the beginning of Section~\ref{sec:time dependent forcing}---, one can view $\LLL \mathbb{P}$ as a compact perturbation of $\LLL$, in which case it would generate a quasi-compact semigroup \cite{engel2000}, meaning that the decay is dominated by a sum of exponential functions, cf. Eq.~\eqref{eq:eigenvalues linear response time dependent}. In case $K(t,s)=K(t-s,0)$, Eq.~\eqref{eq: full response equation} can be solved explicitly in spectral domain and one could relate $K(t-s,0)$ with the high-order responses. Finding further simplifications of the memory kernel $K(t,s)$ would be of great interest, since it would provide a means of computing full responses $R_{\Psi}$ without resorting to perturbation expansions, nor solving the perturbed Fokker-Planck equation.

\small
\newcommand{\etalchar}[1]{$^{#1}$}

\end{document}